\pdfoutput=1 

\documentclass[noprefixes]{cc}

\usepackage{bbm}
\usepackage{xcolor}
\definecolor{darkgreen}{rgb}{0,.35,0}
\definecolor{darkblue}{rgb}{0,0,.5}
\definecolor{darkred}{rgb}{.6,0,0}

\usepackage{hyperref}
\hypersetup{pdfauthor={Mark Giesbrecht, Armin Jamshidpey, \'Eric Schost},
  pdftitle={Subquadratic-Time Algorithms for Normal Bases},
  bookmarksnumbered=true,colorlinks,linkcolor=darkblue,
 citecolor=darkgreen, urlcolor=darkred}
\contact{}% Insert contact email address here.
\submitted{}% Insert the date on which you made the original
\title{Subquadratic-Time Algorithms for Normal Bases
}% Insert title here. (Use \\ to split lines.)
\titlehead{Normal Bases}% If necessary insert running head title here.
\author{
Mark Giesbrecht \\
Cheriton School of Computer\\ Science
University of Waterloo\\ 
Waterloo, ON, Canada  N2L 3G1 \\
\email{mwg@uwaterloo.ca}\\
\homepage{https://cs.uwaterloo.ca/~mwg/}
\and 
Armin Jamshidpey\\
Cheriton School of Computer\\ Science
University of Waterloo\\ 
Waterloo, ON, Canada  N2L 3G1 \\
\hbox to 0pt{\email{armin.jamshidpey@uwaterloo.ca}}
\and 
\'Eric Schost \\
Cheriton School of Computer\\ Science
University of Waterloo\\ 
Waterloo, ON, Canada  N2L 3G1 \\
\email{eschost@uwaterloo.ca}\\
\homepage{https://cs.uwaterloo.ca/~eschost/}\\
}
\begin{abstract}
  For any finite Galois field extension $\K/\F$, with Galois group $G
  = \mathrm{Gal}(\K/\F)$, there exists an element $\alpha \in \K$
  whose orbit $G\cdot\alpha$ forms an $\F$-basis of $\K$. Such an
  $\alpha$ is called a \emph{normal element} and $G\cdot\alpha$ is a
  \emph{normal basis}. We introduce a probabilistic algorithm for
  testing whether a given $\alpha \in \K$ is normal, when $G$ is
  either a finite abelian or a metacyclic group.  The algorithm is
  based on the fact that deciding whether $\alpha$ is normal can be
  reduced to deciding whether $\sum_{g \in G} g(\alpha)g \in \K[G]$ is
  invertible; it requires a slightly subquadratic number of
  operations. Once we know that $\alpha$ is normal, we show how to
  perform conversions between the power basis of $\K/\F$ and the
  normal basis with the same asymptotic cost.
\end{abstract}


\begin{subject}
12Y05, 12F10, 11y16, 68Q25
\end{subject}

\newcommand{\F}{{\mathsf{F}}}
\newcommand{\K}{{\mathsf{K}}}

\def\A{\mathbb{A}}
\def\H{\mathbb{H}}
\def\B{\mathbb{B}}
\def\Z{\mathbb{Z}}
\def\C{\mathbb{C}}
\def\Q{\mathbb{Q}}
\def\D{\mathbb{D}}
\newcommand{\QQ}{\mathbb{Q}}
\newcommand{\mat}[1]{\mathbf{\MakeUppercase{#1}}} % for a matrix

\newcommand{\osumcost}{O(n^{(3/4)\cdot \omega(4/3)})}
\newcommand{\osumcosttilde}{\tilde{O}(n^{(3/4)\cdot \omega(4/3)})}
\newcommand{\thecost}{\tilde{O}(n^{(3/4)\cdot \omega(4/3)})}

\newcommand{\FF}{{\mathbb{F}}}
\newcommand{\xbar}{\xi}

\newcommand{\citeN}{\citet}

\begin{document}

%%% Insert your article text here.

%\begin{acknowledge}
  %% Insert your acknowledgments here.
%\end{acknowledge}
\noacknowledge% This command must be here if there are no
                 % acknowledgements.
                 
\vbox to 0pt{\vspace*{4cm}\tiny\noindent\framebox{To appear: Computational Complexity, 2021}}%

\vspace*{-15pt}
\section{Introduction}

%\vspace*{-10pt}
For a finite Galois field extension $\K/\F$, with Galois group $G =
\mathrm{Gal}(\K/\F)$, an element $\alpha \in \K$ is called
\emph{normal} if the set of its Galois conjugates $G \cdot \alpha = \{
g(\alpha): g\in G\}$ forms a basis for $\K$ as a vector space over
$\F$. The existence of a normal element for any finite Galois extension
is classical, and constructive proofs are provided in most algebra
texts (see, e.g., \citealt[Section 6.13]{Lang}).
% \cite{Wae70}, Section~8.11).
\newpage
 
While there is a wide range of well-known applications of normal bases in
finite fields, such as fast exponentiation~(e.g., \citealt{GaGaPaSh00}), there also
exist applications of normal elements in characteristic zero.  For instance,
in multiplicative invariant theory, for a given permutation lattice and
related Galois extension, a normal basis is useful in computing the
multiplicative invariants explicitly~\citep*{Jam18}.

A number of algorithms are available for finding a normal element in
characteristic zero and in finite fields.  Because of their immediate
applications in finite fields, algorithms for determining normal
elements in this case are most commonly seen.  A fast randomized
algorithm for determining a normal element in a finite field
$\FF_{q^n}/\FF_q$, where $\FF_{q^n}$ is the finite field with $q^n$
elements for any prime power $q$ and integer $n>1$, is presented by
\citeN{GatGie90}, with a cost of $O(n^2+n\log q)$ operations in
$\FF_q$.  A faster randomized algorithm is introduced by
\citeN{KalSho98}, with a cost of $O(n^{1.82}\log q)$ operations in
$\FF_q$.  In the bit complexity model, Kedlaya and Umans showed how to
reduce the exponent of $n$ to $1.63$, by leveraging their quasi-linear
time algorithm for {\em modular
  composition}~\citep{KeUm11}. \cite{LenstraNormal} introduced a
deterministic algorithm to construct a normal element which uses
$n^{O(1)}$ operations in $\FF_{q^n}/\FF_q$.  To the best of our
knowledge, the algorithm of \cite{AugCam94} is the most efficient
deterministic method, with a cost of $O(n^3+n^2\log q)$ operations in
$\FF_q$.

In characteristic zero, \cite{SchSte93} gave an algorithm for finding
a normal basis of a number field over $\QQ$ with a cyclic Galois group
of cardinality $n$ which requires $n^{O(1)}$ operations in $\QQ$.
\cite{Pol94} gives an algorithm for the more general case of finding a
normal basis in an abelian extension $\K/\F$ which requires $n^{O(1)}$
operations in $\F$.  More generally in characteristic zero, for any
Galois extension $\K/\F$ of degree $n$ with Galois group given by a
collection of $n$ matrices, \cite{Girstmair} gives an algorithm which
requires $O(n^4)$ operations in $\F$ to construct a normal element in
$\K$.

In this paper we present a new randomized algorithm that decides
whether a given element in either an abelian or a metacyclic extension
is normal, with a runtime subquadratic in the degree $n$ of the
extension. The costs of all algorithms are measured by counting
\emph{arithmetic operations} in $\F$ at unit cost.  Questions related
to the bit-complexity of our algorithms are challenging, and beyond
the scope of this paper.

Our main conventions are the following.
\begin{assumption}
  \label{assum}
  Let $\K/\F$ be a finite Galois extension presented as
  $\K=\F[x]/\langle P(x)\rangle$, for an irreducible polynomial $P\in
  \F[x]$ of degree $n$, with $\F$ of characteristic zero. Then,
  \begin{itemize}
  \item elements of $\K$ are written on the power basis $1,\xbar,\dots,\xbar^{n-1}$,
    where $\xbar := x \bmod P$;
  \item elements of $G$ are represented by their action on $\xbar$.
  \end{itemize}
\end{assumption}

In particular, for $g \in G$ given by means of $\gamma:=g(\xbar) \in \K$,
and $\beta = \sum_{0\leq i<n}\beta_i\xbar^i\in\K$, the fact that $g$ is an
$\F$-automorphism implies that $g(\beta)$ is equal to $\beta(\gamma)$, the
polynomial composition of $\beta$ (as a polynomial in $\xbar$) at $\gamma$,
reduced modulo $P$.

Our algorithms combine techniques and ideas of~\cite{GatGie90} and
\cite{KalSho98}: $\alpha \in \K$ is normal if and only if the element
$S_\alpha := \sum_{g \in G} g(\alpha)g \in \K[G]$ is invertible in the
group algebra $\K[G]$.  However, writing down $S_\alpha$ involves
$\Theta(n^2)$ elements in $\F$, which precludes a subquadratic
runtime. Instead, knowing $\alpha$, the algorithms use a randomized
reduction to a similar question in $\F[G]$, that amounts to applying a
random projection $\ell:\K\to\F$ to all entries of $S_\alpha$, giving
us an element $s_{\alpha,\ell} \in \F[G]$. For that, we adapt
algorithms from~\citep{KalSho98} that were developed for Galois groups
of finite fields.

Having $s_{\alpha,\ell}$ in hand, we need to test its
invertibility. In order to do so, we present an algorithm in the
abelian case which relies on the fact that $\F[G]$ is isomorphic to a
multivariate polynomial ring modulo an ideal $(x^{e_i}_i-1)_{1 \leq i
  \leq m}$, where $e_i$'s are positive integers. For metacyclic
groups, we exploit the block-Hankel structure of the matrix of
multiplication by $s_{\alpha,\ell}$. 

These latter questions on the cost of arithmetic operations in $\F[G]$
are closely related to that of Fourier transforms over $G$, and it is
worth mentioning that there is a vast literature on fast algorithms
for Fourier transforms (over the base field $\C$). Relevant to our
current context, consider \citep{ClaMu04} and \citep{MaRockWol18} and
references therein for details. At this stage, it is not clear how we
can apply these methods in our context (where we work over an
arbitrary $\F$, not necessarily algebraically closed).

We would like to thank 
the referee for pointing out that~\cite{BuClSh97} address lower and 
upper bounds for multiplying in associative algebras and there is a 
possibility that these results could be applied to our problem of 
testing invertibility in $\F[G]$, though how to do this remains a topic for
future consideration.

This paper is written from the point of view of obtaining improved
asymptotic complexity estimates. Since our main goal is to highlight
the exponent (in $n$) in our runtime analyses, costs are given using
the soft-O notation: $S(n)$ is in $\tilde{O}(T(n))$ if it is in
$O(T(n) \log(T(n))^c)$, for some constant $c$.

The first main result of this paper is the following theorem; we use a
constant $\omega(4/3)$ that describes the cost of certain rectangular
matrix products (see the end of this section).
\begin{theorem}\label{thm:main}
  Under Assumption \ref{assum}, if $G$ is either abelian or
  metacyclic, one can test whether $\alpha \in \K$ is normal using
  $\thecost$ operations in $\F$, where
  $(3/4)\cdot\omega(4/3)<1.99$. The algorithms are randomized of the Monte Carlo type.
\end{theorem}
Once $\alpha$ is known to be normal, we also discuss the cost 
of conversion between the power basis $1,\xbar,\dots,\xbar^{n-1}$
of $\K$ and its normal basis $G\cdot \alpha$. The conversion problem between normal and 
power bases are discussed in~\cite{KalSho98} (randomized) and~\cite{Sergeev} 
(deterministic) with different assumptions. Inspired by
previous work of~\cite{KalSho98}, we obtain the following results.
\begin{theorem}\label{thm:main2}
  Under Assumption \ref{assum}, if $G$ is either abelian or metacyclic
  and $\alpha \in \K$ is known to be normal, we can perform basis
  conversion between the power basis $1,\xbar,\dots,\xbar^{n-1}$ of
  $\K$ and its normal basis $G\cdot \alpha$ using $\thecost$
  operations in $\F$. The algorithms are randomized of the Monte Carlo type.
\end{theorem}
In both theorems, the runtime is barely subquadratic, and the exponent
$1.99$ is obtained through fast matrix multiplication algorithms that
are most likely impractical for reasonable $n$. However, these results
show in particular that we can perform basis conversions without
writing down the normal basis itself (which would require
$\Theta(n^2)$ elements in $\F$).
\begin{remark}\label{rmk:mc-prob}
  Both above algorithms are randomized of the Monte Carlo type. In our
  model, this means that they are allowed to draw random elements for a
  prescribed subset of $\F$, and for a control parameter $\epsilon$,
  produce the correct answer with probability greater than
  $1-\epsilon$ (see Remark \ref{rmk:mc-epsilon}).
\end{remark}

Section \ref{sec:pre} of this paper is devoted to definitions and
preliminary discussions.  In Section \ref{sec:osum}, a
subquadratic-time algorithm is presented for the randomized reduction
of our main question to invertibility testing in $\F[G]$; this
algorithm applies to any finite polycyclic group, and in particular to
abelian and metacyclic groups. In Section \ref{sec:invertibility}, we
show that the problems of testing invertibility in $\F[G]$ and
performing divisions can be solved in quasi-linear time for an abelian
group; for metacyclic groups, we give a subquadratic time algorithm
based on structured linear algebra algorithms (this will finish the
proof of Theorem~\ref{thm:main}). Finally,
Section~\ref{sec:conversion} proves Theorem~\ref{thm:main2}.

Our algorithms make extensive use of known algorithms for polynomial and
matrix arithmetic; in particular, we use repeatedly the fact that
polynomials of degree $n$ in $\F[x]$, for any field $\F$ of
characteristic zero, can be multiplied in $\tilde{O}(n)$ operations in
$\F$~\citep{ScSt71}. As a result, arithmetic operations
$(+,\times,\div)$ in $\K$ can all be done using $\tilde{O}(n)$
operations in $\F$~\citep{vzGathen13}. We also assume that generating a
random element in $\F$ takes constant time.

For matrix arithmetic, we will rely on some non-trivial results on
rectangular matrix multiplication initiated by \cite{LoRo83}. For $k
\in \mathbb{R}$, we denote by $\omega(k)$ a constant such that over
any ring, matrices of sizes $(n,n)$ by $(n,\lceil n^k \rceil)$ can be
multiplied in $O(n^{\omega(k)})$ ring operations (so $\omega(1)$ is
the usual exponent of square matrix multiplication, which we simply
write $\omega$).  The sharpest values known to date for most
rectangular formats are by~\cite{LeGall}; for $k=1$, the best known
value is $\omega \le 2.373$ by \citeN{LeGall14}.  Over a field,
further matrix operations (such as inversion) can also be done in
$O(n^\omega)$ base field operations.

Part of the results of this paper (Theorem~\ref{thm:main} for abelian
groups) were already published in the conference
paper~\citep{GiJaSc19}.

%%% Local Variables:
%%% mode: latex
%%% TeX-master: "NormalBasisCharZero"
%%% End:
 \section{Preliminaries}
\label{sec:pre}

One of the well-known proofs of the existence of a normal element for a
finite Galois extension, as for example reported by \citet[Theorem 6.13.1]{Lang}, suggests a randomized
algorithm for finding such an element. Assume $\K/\F$ is a finite Galois
extension with Galois group $G = \lbrace g_1 , \ldots , g_n \rbrace$. If
$\alpha \in \K$ is a normal element, then
\begin{equation}
  \label{eq:fstrow}
  \sum_{j=1}^n 
  c_j g_j(\alpha)=0, \,\,\, c_j \in \F 
\end{equation} 
implies $c_1 =\dots=c_n = 0$. For each
$i \in \lbrace 1, \ldots , n\rbrace$, applying $g_i$ to equation
\eqref{eq:fstrow} yields
\begin{equation} \label{eq:otherrow} \sum_{j=1}^n c_j g_i g_j(\alpha)=0.
\end{equation}
Using \eqref{eq:fstrow} and \eqref{eq:otherrow}, one can form the linear
system $\mat{M}\boldsymbol{c} = \textbf{0}$,
with $\boldsymbol{c} = [ c_1~ \cdots~c_n]^T$ and
where, for $\alpha\in\K$,
\begin{equation}\label{eqdef:M}
  \mat M =
  \begin{bmatrix}
    g_1 g_1(\alpha) & g_1 g_2(\alpha) & \cdots & g_1 g_n(\alpha) \\
    g_2 g_1(\alpha) & g_2 g_2(\alpha) & \cdots & g_2 g_n(\alpha) \\
    \vdots		& \vdots	& \vdots & \vdots \\
    g_n g_1(\alpha) & g_n g_2(\alpha) & \cdots & g_n g_n(\alpha) \\
  \end{bmatrix} \in M_n(\K).
\end{equation}
Classical proofs then proceed to show that there exists $\alpha \in \K$
with $\det(\mat M)\neq 0$.
 
This approach can be used as the basis of a procedure to test if a
given $\alpha\in\K$ is normal, by computing all the entries of the
matrix $\mat M$ and using linear algebra to compute its determinant;
using fast matrix arithmetic this requires $O(n^\omega)$ operations in
$\K$, that is $\tilde{O}(n^{\omega+1})$ operations in $\F$. This is at
least cubic in $n$; the main contribution of this paper is to show how
to speed up this verification.
 
Before entering that discussion, we briefly comment on the probability
that $\alpha$ be a normal element: if we write $\alpha = a_0 + \cdots
+ a_{n-1} \xbar^{n-1}$, the determinant of $\mat M$ is a
(not identically zero) homogeneous polynomial of degree $n$ in
$(a_0,\dots,a_{n-1})$. If the $a_i$'s are chosen uniformly at random
in a finite set $X \subset \F$, the Lipton-DeMillo-Schwartz-Zippel lemma
implies that the probability that $\alpha$ be normal is at least
$1-n/|X|$.

If $G$ is cyclic generated by an element $g$, with $g_1 = {\rm id}$
and $g_{i+1} = g g_i$ for all $i$, \cite{GatGie90} avoid computing a
determinant by computing the GCD of $S_\alpha := \sum_{i = 1}^{n}
g_i(\alpha)x^{i-1}$ and $x^n-1$. In effect, this amounts to testing
whether $S_\alpha$ is invertible in the group ring $\K[G]$, which is
isomorphic to $\K[x]/\langle x^n-1\rangle$. This is a general fact:
for any $G$, matrix $\mat M$ above is the matrix of left
multiplication by the orbit sum
$$S_\alpha:= \sum_{i=1}^n g_i(\alpha)g_i \in \K[G],$$ where we index
rows by $g_1,\dots,g_n$ and columns by their inverses
$g_1^{-1},\dots,g_n^{-1}$. In terms of notation, for any field
$\mathsf L$ (typically, we will take either $\mathsf L = \F$ or
$\mathsf L = \K$), and $\beta$ in $\mathsf L[G]$, we will write $\mat
M_{\mathsf L}(\beta)$ for the left multiplication matrix by $\beta$ in
$\mathsf L[G]$, using the two bases shown above.  In other words, the
matrix $\mat M$ of~\eqref{eqdef:M} is $\mat M_\K(S_\alpha)$.

The previous discussion shows that $\alpha$ being normal is equivalent
to $S_\alpha$ being a unit in $\K[G]$. This point of view may make it
possible to avoid linear algebra of size $n$ over $\K$, but writing
$S_\alpha$ itself still involves $\Theta(n^2)$ elements in $\F$. The
following lemma is the main new ingredient in our algorithm: it gives
a randomized reduction to testing whether a suitable projection of
$S_\alpha$ in $\F[G]$ is a unit.

\begin{lemma}
  \label{Lem:Proj}
  For $\alpha \in \K$, $\mat M_\K(S_\alpha)$ is invertible if and only
  if $$\ell(\mat M_\K(S_\alpha)) := [\ell(g_ig_j(\alpha))]_{ij} \in M_n(\F)$$
  is invertible for a generic $\F$-linear projection $\ell: \K \to \F$.
\end{lemma}
\begin{proof}
  $(\Rightarrow)$ For a fixed $\alpha\in\K$, any entry of
  $\mat M_\K(S_\alpha)$ can be written as
  \begin{equation}\label{Eq:PrimElm}
    \sum_{k= 0}^{n-1} a_{ijk}\xbar^k, \; a_{ijk} \in F
  \end{equation}
  and for $\ell: \K \to \F$, the corresponding entry in $\ell(\mat
  M_\K(S_\alpha))$ can be written $\sum_{k= 0}^{n-1} a_{ijk}\ell_k$, with
  $\ell_k = \ell(\xbar^k)$. Replacing these $\ell_k$'s by
  indeterminates $L_k$'s, the determinant becomes a polynomial in $P
  \in \F[L_1, \ldots, L_n].$ Viewing $P$ in $\K[L_1, \ldots, L_n]$, we
  have $ P(1, \xbar, \ldots, \xbar^{n-1})$ $= \det(\mat M_\K(S_\alpha))$,
  which is non-zero by assumption. Hence, $P$ is not identically zero,
  and the conclusion follows.
  
  $(\Leftarrow)$ Assume $\mat M_\K(S_\alpha)$ is not invertible. Following the
  proof of \citet[Lemma 4]{Jam18}, we first show that there exists a
  non-zero $\boldsymbol{u} \in \F^n$ in the kernel of $\mat M_\K(S_\alpha)$.
  
  The elements of $G$ act on rows of $\mat M_\K(S_\alpha)$ entrywise and the
  action permutes the rows the matrix. Assume
  $\varphi : G \to \mathfrak{S}_n$ (where $\mathfrak{S}_n$ is the full symmetric group) is the group homomorphism such that
  $g(\mat M_i) = \mat M_{\varphi(g)(i)}$ for all $i$, where $\mat M_i$ is
  the $i$-th row of $\mat M_\K(S_\alpha)$.
  
  Since $\mat M_\K(S_\alpha)$ is singular, there exists a non-zero
  $\boldsymbol{v} \in \K^n$ such that $\mat M_\K(S_\alpha)\boldsymbol{v} = 0$;
  we choose $\boldsymbol{v}$ having the minimum number of non-zero
  entries. Let $i \in \lbrace 1, \ldots , n \rbrace$ such that
  $v_i \neq 0$. Define $\boldsymbol{u} = v_i^{-1}\boldsymbol{v}$. Then,
  $\mat M_\K(S_\alpha)\boldsymbol{u} = 0,$ which means
  $\mat M_j \boldsymbol{u} = 0 $ for $j \in \lbrace 1, \ldots, n
  \rbrace$. For $g \in G$, we have
  $g(\mat M_j \boldsymbol{u}) = \mat M_{\varphi(g)(j)} g(\boldsymbol{u})=
  0.$ Since this holds for any $j$, we conclude that
  $\mat M_\K(S_\alpha)g(\boldsymbol{u})= 0$, hence
  $g(\boldsymbol{u})-\boldsymbol{u}$ is in the kernel of
  $\mat M_\K(S_\alpha)$. On the other hand since the $i$-th entry of
  $\boldsymbol{u}$ is one, the $i$-th entry of
  $g(\boldsymbol{u}) -\boldsymbol{u}$ is zero. Thus the minimality
  assumption on $\textbf{v}$ shows that
  $g(\boldsymbol{u}) -\boldsymbol{u} = 0$, equivalently
  $g(\boldsymbol{u})=\boldsymbol{u}$, and hence $\boldsymbol{u} \in \F^n$.
  
  Now we show that for all choices of $\ell$, $\ell(\mat M_\K(S_\alpha))$ is not invertible. By Equation \eqref{Eq:PrimElm}, we can write
  $$\mat M_\K(S_\alpha) = \sum_{j = 0}^{n-1} \mat M^{(j)} \xbar^j, \quad 
  \mat M^{(j)} \in M_{n}(\F) \text{~for all $j$}.$$ 
  Since $\boldsymbol{u}$ has entries in $\F$,
  $\mat M_\K(S_\alpha) \boldsymbol{u} =0$ yields
  $\mat M^{(j)}\boldsymbol{u} = 0$ for
  $j \in \lbrace 1, \ldots , n \rbrace$. Hence,
$$\sum_{j = 0}^{n-1} \mat M^{(j)} \ell_j \boldsymbol{u} = 0$$ for any 
$\ell_j$'s in $\F$, and $\ell(\mat M_\K(S_\alpha))$ is not invertible for any~$\ell$.
\end{proof} 
Our algorithm can be sketched as follows: given $\alpha$ in $\K$,
choose a random $\ell: \K\to\F$, and let
\begin{equation}\label{def:s_alpha_ell}
s_{\alpha,\ell}:=\sum_{i=1}^n \ell(g_i(\alpha))g_i \in \F[G].
\end{equation}
Note that $\ell(\mat M_\K(S_\alpha))$ is equal to $\mat
M_\F(s_{\alpha,\ell})$, that is, the multiplication matrix by
$s_{\alpha,\ell}$ in $\F[G]$, where, as above, we index rows by
$g_1,\dots,g_n$ and columns by $g_1^{-1},\dots,g_n^{-1}$.  Then,
the previous lemma can be rephrased as follows:
\begin{lemma}
  \label{Lem:Proj-bis}
  For $\alpha \in \K$, $\alpha$ is normal if and only if
  $s_{\alpha,\ell}$ is invertible in $\F[G]$ for a generic
  $\F$-linear projection $\ell: \K \to \F$.
\end{lemma}
Thus, once $s_{\alpha,\ell}$ is known, we are left with testing
whether it is a unit in $\F[G]$. In the next two sections, we address
the respective questions of computing $s_{\alpha,\ell}$, and testing
its invertibility in $\F[G]$.
\begin{remark}\label{rmk:mc-epsilon}
If $\alpha$ is not normal, $S_\alpha$ is not a unit. In this case, the
proof of Lemma \ref{Lem:Proj} established that $s_{\alpha,\ell}$ is
not a unit for {\em any} $\ell$, so our algorithm always returns the
correct answer in this case.

If $\alpha$ is normal, the polynomial $P$ in the proof of Lemma
\ref{Lem:Proj}, is a homogeneous polynomial of degree $n$ in $(L_1,
\ldots , L_n)$. Thus, if we choose the coefficients of $\ell$ uniformly at
random in any fixed finite subset $X \subset \F$, by the
Lipton-DeMillo-Schwartz-Zippel lemma, we return the correct 
answer with probability at least $1-n/|X|$. 

By Lemma \ref{Lem:Proj}, if $\ell(\mat M_\K(S_\alpha))$ is invertible, then $\alpha$ is normal. 
As one of the referees pointed out, one can use this as a Las Vegas 
algorithm which returns a normal element.
\end{remark}

%%% Local Variables:
%%% mode: latex
%%% TeX-master: "NormalBasisCharZero"
%%% End:
 \section{Computing projections of the orbit sum}
\label{sec:osum}

In this section we present an algorithm to compute $s_{\alpha,\ell}$
when $G=\{g_1,\dots,g_n\}$ is {\em polycyclic} (we give a definition
of this family of groups and recall some well known results about them
in Subsection~\ref{ssec:proj_abelian}). To motivate our algorithm,
we start by the simple case of a {\em cyclic} group.  We will see that
they follow closely ideas used by \cite{KalSho98} over finite fields.

Suppose $G = \langle g \rangle$, so that given $\alpha$ in $\K$ and
$\ell: \K \to \F$, our goal is to compute
\begin{equation}
  \label{eq:cycproj}
  \ell(g^i(\alpha)), ~~\mbox{for}~ 0\leq i\leq n-1.
\end{equation}
\cite{KalSho98} call this the \emph{automorphism projection problem} and
gave an algorithm to solve it in subquadratic time, when $g$ is the
$q$-power Frobenius $\mathbb{F}_{q^n} \to \mathbb{F}_{q^n}$.  The key idea in their
algorithm is to use the baby-steps/giant-steps technique: for a suitable
parameter $t$, the values in \eqref{eq:cycproj} can be rewritten as
\[
  (\ell \circ g^{tj})(g^i(\alpha)), ~~\mbox{for}~ 0 \leq j < m:=\lceil n/t
  \rceil ~\mbox{and}~ 0 \leq i <t.
\]
First, we compute all $G_i:=g^i(\alpha)$ for $0 \leq i <t$.  Then we compute
all $L_j:=\ell \circ g^{tj}$ for $0 \leq j <m$, where the $L_j$'s are
themselves linear mappings $\K \to \F$.  Finally, a matrix product yields
all values $L_j(G_i)$.

The original algorithm of \cite{KalSho98} relies on the properties of
the Frobenius mapping to achieve subquadratic runtime. In our case, we
cannot apply these results directly; instead, we have to revisit the
proofs of~\cite[Lemmata 3 and 4]{KalSho98}, now considering
rectangular matrix multiplication.  Our exponents involve the constant
$\omega(4/3)$, for which we have the upper bound $\omega(4/3) <
2.654$: this follows from the upper bounds on $\omega(1.3)$ and
$\omega(1.4)$ given by~\cite{LeGall}, and the fact that $k \mapsto
\omega(k)$ is convex~\citep{LoRo83}. In particular, $3/4 \cdot
\omega(4/3) < 1.99$. Note also the inequality $\omega(k) \ge 1+k$ for
$k\ge 1$, since $\omega(k)$ describes products with input and output
size $O(n^{1+k})$.

%%%%%%%%%%%%%%%%%%%%%%%%%%%%%%%%%%%%%%%%%%%%%%%%%%%%%%%%%%%%

\subsection{Multiple automorphism evaluation and applications}

The key to the algorithms below is the remark following
Assumption~\ref{assum}, which reduces automorphism evaluation to
modular composition of polynomials.  Over finite fields, this idea
goes back to~\cite{GaSh92}, where it is credited to Kaltofen.

For instance, given $g \in G$ (by means of $\gamma:=g(\xbar)$), we can
deduce $g^2 \in G$ (again, by means of its image at $\xbar$) as
$\gamma(\gamma)$; this can be done with $\tilde{O}(n^{(\omega+1)/2})$
operations in $\F$ using Brent and Kung's modular composition
algorithm~\citep{BrKu78}. The algorithms below describe similar
operations along these lines, involving several simultaneous
evaluations. In this subsection, we work under Assumption~\ref{assum}
and we make no special assumption on $G$.

\begin{lemma}
  \label{lem:modcom}
  Given $\alpha_1,\dots,\alpha_s$ in $\K$ and $g$ in $G =
  \mathrm{Gal}(\K/\F)$, with $s = O(\sqrt{n})$, we can compute
  $g(\alpha_1),\dots,g(\alpha_s)$ with $\tilde
  O(n^{(3/4)\cdot\omega(4/3)})$ operations in $\F$.
\end{lemma}
\begin{proof}
(Compare \citealt[Lemma~3]{KalSho98}) As noted above, for $i\le s$,
  $g(\alpha_i) = \alpha_i(\gamma)$, with $\gamma := g(\xbar) \in \K$.
  Let $t := \lceil n^{3/4} \rceil$, $m:=\lceil n/t\rceil$, and rewrite $\alpha_1 , \ldots , \alpha_s$ as 
$$\alpha_i = \sum_{0 \leq j < m} a_{i,j}\xbar^{tj},$$ where the
  $a_{i,j}$'s are polynomials of degree less than $t$. The next step
  is to compute $\gamma_i := \gamma^i$, for $i = 0 , \ldots , t$.
  There are $t$ products in $\K$ to perform, so this amounts to
  $\tilde{O}(n^{7/4})$ operations in $\F$.

  Having $\gamma_i$'s in hand, one can form the matrix
  $\boldsymbol{\Gamma} := \left[ \Gamma_0 ~ \cdots ~ \Gamma_{t-1}
    \right]^T$, where each column $\Gamma_i$ is the coefficient vector
  of $\gamma_i$ (with entries in $\F$); this matrix has $t \in
  O(n^{3/4})$ rows and $n$ columns. We also form
  $$\mat A := \left[{A}_{1,0} \cdots {A}_{1,m-1} \cdots
    {A}_{s,0} \cdots {A}_{s,m-1}\right]^T,$$ where
  ${A}_{i,j}$ is the coefficient vector of $a_{i,j}$. This matrix 
  has $s m \in O(n^{3/4})$ rows and $t \in O(n^{3/4})$ columns.

  Compute $\mat B:=\mat A\, \boldsymbol{\Gamma}$; as per our
  definition of exponents $\omega(\cdot )$, this can be done in
  $O(n^{(3/4)\cdot \omega(4/3)})$ operations in $\F$, and the rows of this matrix
  give all $a_{i,j}(\gamma)$.  The last step to get all
  $\alpha_i(\gamma)$ is to write them as $\alpha_i(\gamma) = \sum_{0
    \leq j < m} a_{i,j}(\gamma) \gamma_t^{j}.$ Using Horner's scheme,
  this takes $O(sm)$ operations in $\K$, which is $\tilde{O}(n^{7/4})$
  operations in $\F$. Since $(3/4)\cdot\omega(4/3) \ge 7/4$,
  the leading exponent in all costs seen so far is
  $(3/4)\cdot\omega(4/3)$.
\end{proof}

\begin{lemma}\label{lem:selfcomp}
  Consider $g_1, \ldots g_r$ in $G=\mathrm{Gal}(\K/\F)$, positive
  integers $(s_1, \ldots s_r)$ and elements $\alpha_{i_1, \dots, i_r}$
  in $\K$, for $i_m=0,\dots,s_m$, $m=1,\dots,r$. If $\prod_{i = 1}^r
  s_i = O(\sqrt{n})$ and $r = O(\log(n))$, we can compute
  $$g_r^{i_r}\cdots g_1^{i_1}(\alpha_{i_1, \dots, i_r})
  \text{~for~} i_m=0,\dots,s_m,\ m=1,\dots,r
  $$  using $\thecost$ operations in $\F$.
\end{lemma}
\begin{proof}
  Define $\mathcal I = \{(i_1,\dots,i_r) \mid i_m=0,\dots,s_m
  \text{~for~} m=1,\dots,r\}$. For $(i_1,\dots,i_r)$ in $\mathcal I$
  and non-negative integers $\ell_1,\dots,\ell_r$, define
  $$\alpha_{i_1, \dots, i_r}^{(\ell_1,\dots,\ell_r)} 
  =g_r^{\ell_r} \cdots g_1^{\ell_1}(\alpha_{i_1, \dots, i_r}).$$  
  Assume then that
  for some $t$ in $\{0,\dots,r-1\}$, we know 
  $$S_t=(\alpha_{i_1, \dots,i_r}^{(i_1, \dots, i_{t},0, \dots, 0)} \ \mid \ (i_1,\dots,i_r)\in \mathcal I);$$ we show how to compute 
  $$S_{t+1}=(\alpha_{i_1, \dots,i_r}^{(i_1, \dots, i_{t+1},0, \dots,
    0)} \ \mid \ (i_1,\dots,i_r)\in \mathcal I).$$ Since our input is
  $S_0$, it will be enough to go through this process for all values
  of $t$ to obtain the output $S_r$ of the algorithm.
  
  For a given index $t$, and for $m \ge 0$
  define further 
  $$S_{t,m}=(\alpha_{i_1, \dots,i_r}^{(i_1, \dots, i_{t},i_{t+1} \bmod
    2^m, 0,\dots, 0)} \ \mid \ (i_1,\dots,i_r)\in \mathcal I);$$ in
  particular, $S_{t,0} = S_t$ and $S_{t,\lfloor
    \log_2(s_{t+1})\rfloor+1} = S_{t+1}$.  Hence, given $S_{t,m}$, it
  is enough to show how to compute $S_{t,m+1}$, for indices
  $m=0,\dots,\lfloor \log_2(s_{t+1})\rfloor$.  This is done by writing
  $$S_{t,m+1}=
  (\beta_{i_1, \dots,i_r,t,m} \ \mid \ (i_1,\dots,i_r)\in \mathcal I),$$
  with
  $$\beta_{i_1, \dots,i_r,t,m}
  =\begin{cases}
  \alpha_{i_1, \dots,i_r}^{(i_1, \dots, i_{t},i_{t+1} \bmod 2^m,0, \dots, 0)}  \text{~if~} 
  i_{t+1} \bmod 2^{m+1} = i_{t+1} \bmod 2^m \\[2mm]
  g_{t+1}^{2^{m}}(\alpha_{i_1, \dots,i_r}^{(i_1, \dots, i_{t},i_{t+1} \bmod 2^m,0, \dots, 0)}) 
  \text{~otherwise.}
  \end{cases}$$
  The automorphisms $g_{t+1}^{2^m}$ can be computed iteratively by modular
  composition; the bottleneck is the application of  $g_{t+1}^{2^m}$
  to a subset of $S_{t,m}$. Using Lemma \ref{lem:modcom}, since 
  $S_{t,m}$ has $O(\sqrt n)$ elements, this takes $\thecost$ 
  operations in $\F$.
  
  For a given index $t$, this is repeated $\lfloor
  \log_2(s_{t+1})\rfloor \le \log_2(s_{t+1}) + 1$ times. Adding up for
  all indices $t$, this amounts to $O(\log (s_1 \cdots s_r)+r)$
  repetitions, which is $O(\log(n))$ by assumption; the conclusion
  follows.
\end{proof}

We now present dual versions of the previous two lemmas (note that
\cite{KalSho98} also have such a discussion). Seen as an $\F$-linear
map, the operator $g:\alpha \mapsto g(\alpha)$ admits a transpose,
which maps an $\F$-linear form $\ell:\K\to\F$ to the $\F$-linear form
$\ell \circ g: \alpha \mapsto \ell(g(\alpha))$.  The {\em
  transposition principle}~\citep{KaKiBs88,CaKaYa89} implies that if a
linear map $\F^N \to \F^M$ can be computed in time $T$, its transpose
can be computed in time $T+O(N+M)$. In particular, given $s$ linear
forms $\ell_1,\dots,\ell_s$ and $g$ in $G$, transposing
Lemma~\ref{lem:modcom} shows that we can compute $\ell_1 \circ
g,\dots,\ell_s \circ g$ in time $\osumcosttilde$. The following lemma
sketches the construction.

\begin{lemma}
  \label{lem:modcomT}
  Given $\F$-linear forms $\ell_1,\dots,\ell_s:\K\to \F$ and $g$ in $G =
  \mathrm{Gal}(\K/\F)$, with $s = O(\sqrt{n})$, we can compute
  $\ell_1\circ g,\dots,\ell_s \circ g$ 
 using $\thecost$ operations in $\F$.
\end{lemma}
\begin{proof}
  Given $\ell_i$ by its values on the power basis $1,\xbar,\dots,\xbar^{n-1}$, $\ell_i \circ g$ is represented by its values at
  $1,\gamma,\dots,\gamma^{n-1}$, with $\gamma := g(\xbar)$. 

  Let $t,m$ and $\gamma_0,\dots,\gamma_t$ be as in the proof of
  Lemma~\ref{lem:modcom}. Compute the ``giant steps''
  $\gamma_t^j = \gamma^{tj}$, $j=0,\dots,m-1$ and for $i=1,\dots,s$
  and $j=0,\dots,m-1$, deduce the linear forms $L_{i,j}$ defined by
  $L_{i,j}(\alpha) := \ell_i(\gamma^{tj}\alpha)$ for all $\alpha$ in
  $\K$. Each of them can be obtained by a {\em transposed
    multiplication} in time $\tilde{O}(n)$~\citep[Section~4.1]{Shoup},
  so that the total cost thus far is $\tilde{O}(n^{7/4})$.

  Finally, multiply the $(sm \times n)$ matrix with entries the
  coefficients of all $L_{i,j}$ (as rows) by the $(n \times t)$ matrix with
  entries the coefficients of $\gamma_0,\dots,\gamma_{t-1}$ (as columns) to
  obtain all values $\ell_i(\gamma^j)$, for $i=1,\dots,s$ an
  $j=0,\dots,n-1$.  This can be accomplished with
  $O(n^{(3/4)\cdot\omega(4/3)})$ operations in~$\F$.
\end{proof}

From this, we deduce the transposed version of Lemma~\ref{lem:selfcomp},
whose proof follows the same pattern.

\begin{lemma}
  \label{lem:transmodcomp}
  Consider $g_1, \ldots, g_r$ in $G=\mathrm{Gal}(\K/\F)$, positive
  integers $(s_1, \ldots, s_r)$ and $\F$-linear forms $\ell_{i_1,
    \dots, i_r}$, for $i_m=0,\dots,s_m$, $m=1,\dots,r$. If $\prod_{i =
    1}^r s_i = O(\sqrt{n})$ and $r = O(\log(n))$, we can compute
  $$\ell_{i_1, \dots, i_r} \circ g_r^{i_r}\cdots g_1^{i_1}
  \text{~for~} i_m=0,\dots,s_m,\ m=1,\dots,r
  $$  using $\thecost$ operations in $\F$.
\end{lemma} 
\begin{proof}
  We proceed as in Lemma~\ref{lem:selfcomp}, reversing the order of
  the steps. Using the same index set $\mathcal I$ as before, define,
  for $(i_1,\dots,i_r)$ in $\mathcal I$ and non-negative integers
  $k_1,\dots,k_r$
  $$\ell_{i_1,\dots,i_r}^{(k_1,\dots,k_r)} =\ell_{i_1,\dots,i_r} \circ g_r^{k_r}\cdots g_1^{k_1}.$$
  For $t=r,\dots,0$, assuming that
  we know 
  $$L_{t+1} = (\ell_{i_1, \dots,i_r}^{(0, \dots, 0,i_{t+1},\dots,i_r)} \ \mid
  \ (i_1,\dots,i_r)\in \mathcal I),$$ we compute 
  $$L_{t}=(\ell_{i_1, \dots,i_r}^{(0, \dots, 0,i_{t},i_{t+1},\dots,i_r)}
  \ \mid \ (i_1,\dots,i_r)\in \mathcal I).$$
  This time, for $m \ge 0$, we set
  $$L_{t+1,m} = (\ell_{i_1, \dots,i_r}^{(0, \dots, 0,\lfloor i_{t}
    \rfloor_m,i_{t+1},\dots,i_r)} \ \mid \ (i_1,\dots,i_r)\in \mathcal
  I),$$ where for a non-negative integer $x$, $\lfloor x \rfloor_m = x
  - (x \bmod (2^m-1))$ is obtained by setting to zero the coefficients
  of $1,2,\dots,2^{m-1}$ in the base-two expansion of $x$.

  Starting from $L_{t+1} = L_{t, \lceil \log_2(s_t) \rceil +1}$, we
  compute all $L_{t+1,m}$ for $m= \lceil \log_2(s_t) \rceil,\dots,0$,
  since $L_{t+1,0} = L_{t}$. This is done essentially as in
  Lemma~\ref{lem:selfcomp}, but using Lemma~\ref{lem:modcomT} this
  time, in order to do right-composition by $g_t^{2^m}$.
  The cost analysis is as in Lemma~\ref{lem:selfcomp}.
\end{proof}

%%%%%%%%%%%%%%%%%%%%%%%%%%%%%%%%%%%%%%%%%%%%%%%%%%%%%%%%%%%%

\subsection{Computing the orbit sum projection for polycyclic groups}
\label{ssec:proj_abelian}

Our main algorithm in this section applies to a family of groups known
as {\em polycyclic}; see~\cite[Chapter 8]{HoEiOb05} for more details
on such groups.

Our group $G$ is called polycyclic if it has a normal series
$$G = G_{r} \unrhd G_{r-1} \unrhd \cdots \unrhd G_1 \unrhd G_{0} =
1,$$ where $G_{j}/G_{j-1}$ is cyclic; without loss of generality, we
assume that $G_{j-1 } \ne G_{j}$ holds for all $j$, so that $r$ is
$O(\log(n))$, with $n=|G|$. Finitely generated nilpotent or abelian
groups are polycyclic. In general any finite solvable group is
polycyclic; our key families of examples in the next section (abelian
and metacyclic groups) thus fit into this category.

If $G$ is polycyclic then, up to renumbering, its 
elements can be written as
\[\label{eq:polycyclicgrp}
g_r^{i_r} \cdots g_1^{i_1}, \text{~with~} 0 \leq i_j < e_j
\text{~for~} 1 \leq j \leq r,\
\]
where $G_{j}/G_{j-1} = \langle g_{j}G_{j-1} \rangle$ and $e_j=\vert
G_{j}/G_{j-1}\vert$. Elements of $\K[G]$, or $\F[G]$ are written as
polynomials $\sum_{i_1,\dots,i_r} c_{i_1,\dots,i_r} {g_r}^{i_r} \cdots
{g_1}^{i_1}$, with $0\le i_j < e_j$ for all~$j$, and coefficients 
in either $\K$ or $\F$.

%% Assume elements of a polycyclic group $G$ are presented as in
%% \eqref{eq:polycyclicgrp}, so that elements of $\F[G]$ are written as
%% polynomials $$\sum_{i_1,\dots,i_r} c_{i_1,\dots,i_r} {g_r}^{i_r}
%% \cdots {g_1}^{e_1},$$ with $0\le i_j < e_j$ for all~$j$.

\begin{proposition}\label{prop:polycyclic}
  Suppose that $G$ is polycyclic, with notation as above. For $\alpha$ in
  $\K$ and $\ell:\K\to\F$, $s_{\alpha,\ell} \in \F[G]$, as defined
  in~\eqref{def:s_alpha_ell}, can be computed using $\osumcosttilde$
  operations in $\F$.
\end{proposition}
\begin{proof}
Our goal is to compute
\begin{equation} \label{eq:polycyclic}
  \ell (g_r^{i_r}  \ldots g_1^{i_1}(\alpha)),
\end{equation}
for all indices such that $0 \leq i_j < e_j$ holds for $1 \leq j \leq
r$; here, $\ell$ is an $\F$-linear projection $\K\to \F$.

Our construction is inspired by that sketched in the cyclic case.
Define $z$ to be the unique index in $\{1,\dots,r\}$ such that
$e_1\cdots e_{z-1} < \sqrt{n}$ and $e_1 \cdots e_{z} \geq
\sqrt{n}.$ Then, all elements in \eqref{eq:polycyclic} can be computed
with the following steps, the sum of whose costs proves the
proposition.

\smallskip\noindent \textbf{Step 1.} Apply Lemma \ref{lem:selfcomp},
with $\alpha_{i_1,\dots,i_r} = \alpha$ for all $i_1,\dots,i_r$, to get
$$ G_{i_z,\dots,i_1}=g_z^{i_z} \cdots
g_1^{i_1}(\alpha),$$ for all indices $i_1,\dots,i_z$ such that $0\leq
i_m < e_m$ holds for $m=1,\dots,z-1$ and $0\leq i_z <
\lceil{\sqrt{n}}/(e_1 \cdots e_{z-1})\rceil$.  This amounts to taking
$s_1=e_1,\dots,s_{z-1}=e_{z-1}$, $s_z=\lceil {\sqrt{n}}/(e_1 \cdots
e_{z-1})\rceil$ and $s_m = 1$ for $m > z$ in the lemma.  For the lemma to
apply, we have to check that the product of these indices
$s_1,\dots,s_r$ is $O(\sqrt n)$. 
Indeed, this product is at most 
\begin{align*}
e_1\cdots e_{z-1} \left ( \frac{\sqrt n}{e_1\cdots e_{z-1}} + 1\right )
& \le \sqrt n + e_1\cdots e_{z-1} \le 2 \sqrt n.
\end{align*}
Hence, the lemma applies, and the cost of this step is $\osumcosttilde$.

\smallskip\noindent\textbf{Step 2.} Compute $G_z =g_z^{s_z}$, for $s_z$
as above. The cost is that of $O(\log(n))$ modular compositions, which
is negligible compared to the cost of the previous step.

\smallskip\noindent\textbf{Step 3.} Use Lemma \ref{lem:transmodcomp}
with $\ell_{i_r,\dots,i_1} = \ell$ for all $i_1,\dots,i_r$, to compute 
\begin{align*}
L_{j_r,\dots,j_z} &= \ell \circ (g_{r}^{j_{r}}\cdots g_{z+1}^{j_{z+1}}G_z^{j_z})\\ 
&= \ell \circ (g_{r}^{j_{r}}\cdots g_{z+1}^{j_{z+1}}g_{z}^{s_z j_z}),
\end{align*}
for all indices $0 \leq j_{z} < \lceil {e_{z}}/{s_z}\rceil$ and $0
\leq j_m < e_m$ for $m > z$. This amounts to using the lemma with
indices $s'_1=\cdots=s'_{z-1}=1$, $s'_z = \lceil {e_{z}}/{s_z}\rceil$
and $s'_m = e_m$ for $m > z$. Again, we have to verify that $s'_1
\cdots s'_r$ is $O(\sqrt n)$.  Indeed, we have
\begin{align*}
s'_1 \cdots s'_r  = \left \lceil \frac{e_{z}}{s_z}\right \rceil e_{z+1} \cdots e_r
&\le \left (\frac{e_{z}}{s_z} +1 \right) e_{z+1} \cdots e_r\\
&\le \frac{e_{z} \cdots e_r}{s_z} + e_{z+1} \cdots e_r.
\end{align*}
By definition, we have $s_z \ge \sqrt{n}/(e_1\cdots e_{z-1})$, so
$e_z \cdots e_r/s_z \le e_1 \cdots e_r /\sqrt{n} =\sqrt{n}$.
Because we assume $e_1 \cdots e_{z} \geq \sqrt{n}$, the second term is
also at most $\sqrt{n}$, so the product $s'_1 \cdots s'_r$ is at most
$2 \sqrt{n}$. Hence, Lemma \ref{lem:transmodcomp} applies, and computes
all $L_{j_r,\dots,j_z}$ using $\osumcosttilde$ operations
in $\F$.

\smallskip\noindent\textbf{Step 4.} Multiply the matrix with rows the
coefficients of all $L_{j_r,\dots,j_z}$ by the matrix whose columns
are the coefficients of all $G_{i_z,\dots,i_1}$. This yields the
values
$$\ell( 
g_{r}^{j_{r}}\cdots g_{z+1}^{j_{z+1}}g_{z}^{s_z j_z+i_z} g_{z-1}^{i_{z-1}}\cdots
g_1^{i_1}(\alpha)),$$
for indices as follows:
\begin{itemize}
\item[$\bullet$] $0 \le i_m < e_m$ for $m=0,\dots,z-1$;
\item[$\bullet$] $0 \le i_z < s_z$ and $0 \le j_z< \lceil  e_z/s_z\rceil$;
\item[$\bullet$] $0 \le j_m < e_m$ for $m=z+1,\dots,r$.
\end{itemize}
This shows that we obtain all required values. We compute this product
in $O(n^{(1/2)\cdot\omega(2)})$ operations in $\F$, which is in
$\osumcost$.
\end{proof}

%%% Local Variables:
%%% mode: latex
%%% TeX-master: "NormalBasisCharZero"
%%% End:
 \section{Arithmetic in the Group Algebra}
\label{sec:invertibility}

In this section we consider the problems of invertibility testing and
division in $\F[G]$: given elements $\beta,\eta$ in $\F[G]$, for a
field $\F$ and a group $G$, determine whether $\beta$ is a unit in
$\F[G]$, and if so, compute $\beta^{-1}\eta$. We focus on two
particular families of polycyclic groups, namely abelian and
metacyclic groups $G$; as well as being necessary in our application
to normal bases, we believe these problems are of independent
interest.

Since we are in characteristic zero, Wedderburn's theorem implies the
existence of an $\F$-algebra isomorphism (which we will refer to as a
Fourier Transform)
\[
  \F[G] \to M_{d_1}(D_1) \times \dots \times M_{d_r}(D_r),
\]
where all $D_i$'s are division algebras over $\F$. If we were working
over $\F=\C$, all $D_i$'s would simply be $\C$ itself.  A natural
solution to test the invertibility of $\beta \in \F[G]$ would then be
to compute its Fourier transform and test whether all its components
$\beta_1 \in M_{d_1}(\C),\dots,\beta_r \in M_{d_r}(\C)$ are
invertible. This boils down to linear algebra over $\C$, and takes
$O(d_1^\omega + \cdots + d_r^\omega)$ operations.  Since $d_1^2 +
\cdots + d_r^2 = n$, with $n=|G|$, this is $O(n^{\omega/2})$
operations in $\C$.

However, we do not wish to make such a strong assumption as $\F=\C$. Since
we measure the cost of our algorithms in $\F$-operations, the direct
approach that embeds $\F[G]$ into $\C[G]$ does not make it possible to
obtain a subquadratic cost in general. If, for instance, $\F=\Q$ and $G$ is
cyclic of order $n=2^k$, computing the Fourier Transform of $\beta$
requires we work in a degree $n/2$ extension of $\Q$, implying a quadratic
runtime.

In this section, we give algorithms for the problems of invertibility
testing and division for the two particular families of polycyclic
groups mentioned so far, namely abelian and metacyclic. For the former,
starting from a suitable presentation of $G$, we give a softly
linear-time algorithm to find an isomorphic image of $\beta \in \F[G]$
in a product of $\F$-algebras of the form $\F[z]/\langle
P_i(z)\rangle$, for certain polynomials $P_i \in \F[z]$ (recovering
$\beta$ from its image is softly-linear time as well). Not only does
this allow us to test whether $\beta$ is invertible, this also makes
it possible to find its inverse in $\F[G]$ (or to compute products in
$\F[G]$) in softly-linear time (we are not aware of previous results
of this kind).

For metacyclic groups, we rely on the block-Hankel structure of the
matrix of multiplication by $\beta$. Through structured linear algebra
algorithms, this allows us to solve both problems (invertibility and
division) in subquadratic (albeit not softly-linear time) time.

%%%%%%%%%%%%%%%%%%%%%%%%%%%%%%%%%%%%%%%%%%%%%%%%%%%%%%%%%%%%

\subsection{Abelian groups}

Because an abelian group is a product of cyclic groups, the group
algebra $\F[G]$ of such a group is the tensor product of cyclic
algebras. Using this property, given an element $\beta$ in $\F[G]$,
our goal in this section is to determine whether $\beta$ is a unit,
and if so to compute expressions such as $\beta^{-1} \eta$, for
$\eta$ in $\F[G]$.

The previous property implies that $\F[G]$ admits a description of the
form $\F[x_1,\dots,x_t]/\langle x_1^{n_1}-1,\dots,x_t^{n_t}-1\rangle$,
for some integers $n_1,\dots,n_t$. The complexity of arithmetic
operations in an $\F$-algebra such as $\A:=\F[x_1,\dots,x_t]/\langle
P_1(x_1),\dots,P_t(x_t)\rangle$ is difficult to pin down precisely. For
general $P_i$'s, the cost of multiplication in $\A$ is known to be
$O(\dim(\A)^{1+\varepsilon})$, for any $\varepsilon >
0$~\citep[Theorem~2]{LiMoSc09}. From this it may be possible to deduce
similar upper bounds on the complexity of invertibility test or division,
following~\citep{DaMMMScXi06}, but this seems non-trivial.

Instead, we give an algorithm with softly linear runtime, that uses
the factorization properties of cyclotomic polynomials and Chinese
remaindering techniques to transform our problem into that of
invertibility test or division in algebras of the form $\F[z]/\langle
P_i(z) \rangle$, for various polynomials $P_i$.  \cite{Pol94} also
discusses the factors of algebras such as $\F[x_1,\dots,x_t]/\langle
x_1^{n_1}-1,\dots,x_t^{n_t}-1\rangle$, but the resulting algorithms
are different (and the cost of the \citeauthor{Pol94}'s
\citeyearpar{Pol94} algorithm is only known to be polynomial in
$n=|G|$).

\smallskip

\noindent{\bf Tensor product of two cyclotomic rings: coprime orders.}
The following proposition will be the key to foregoing multivariate
polynomials, and replacing them by univariate ones.  Let $m,m'$ be two
coprime integers and define
$$\mathbbm{h}:=\F[x,x']/\langle \Phi_{m}(x), \Phi_{m'}(x')\rangle,$$
where for $i \ge 0$, $\Phi_i$ is the cyclotomic polynomial of order
$i$. In what follows, $\varphi$ is Euler's totient function, so that
$\varphi(i) = \deg(\Phi_i)$ for all~$i$.
\begin{lemma}
  There exists an $\F$-algebra isomorphism $\gamma: \mathbbm{h} \to
  \F[z]/\langle\Phi_{mm'}(z)\rangle$ given by $xx' \mapsto z$.  Given
  $\Phi_m$ and $\Phi_{m'}$, $\Phi_{mm'}$ can be computed in time
  $\tilde{O}(\varphi(mm'))$; given these polynomials, one can
  apply $\gamma$ and its inverse to any input using
  $\tilde{O}(\varphi(mm'))$ operations in~$\F$.
\end{lemma}
\begin{proof}
  Without loss of generality, we prove the first claim over $\Q$; the
  result over $\F$ follows by scalar extension. In the field \sloppy
  $\Q[x,x']/\langle \Phi_{m}(x), \Phi_{m'}(x')\rangle$, $xx'$ is
  cancelled by $\Phi_{mm'}$. Since this polynomial is irreducible, it
  is the minimal polynomial of $xx'$, which is thus a primitive
  element for $\Q[x,x']/\langle \Phi_{m}(x),
  \Phi_{m'}(x')\rangle$. This proves the first claim.

  For the second claim, we first determine the images of $x$ and $x'$
  by $\gamma$. Start from a B\'ezout relation $am+ a'm'=1$, for some
  $a,a'$ in $\Z$.  Since $x^m = {x'}^{m'}=1$ in $\mathbbm{h}$, we
  deduce that $\gamma(x)=z^{u}$ and $\gamma(x') = z^{v}$, with $u:=am
  \bmod mm'$ and $v:=a'm' \bmod mm'$. To compute $\gamma(P)$, for some
  $P$ in $\mathbbm{h}$, we first compute $P(z^u, z^v)$, keeping all
  exponents reduced modulo $mm'$. This requires no arithmetic
  operations and results in a polynomial $\bar P$ of degree less than
  $mm'$, which we eventually reduce modulo $\Phi_{mm'}$ (the latter is
  obtained by the composed product algorithm of~\cite{BoFlSaSc06} in
  quasi-linear time).  By~\citep[Theorem~8.8.7]{BacSha96}, we have the
  bound $s \in O(\varphi(s) \log(\log(s)))$, so that $s$ is in
  $\tilde{O}(\varphi(s))$. Thus, we can reduce $\bar P$ modulo
  $\Phi_{mm'}$ in $\tilde{O}(\varphi(mm'))$ operations, establishing
  the cost bound for $\gamma$.

  Conversely, given $Q$ in $\F[z]/\langle\Phi_{mm'}(z)\rangle$, we obtain
  its preimage by replacing powers of $z$ by powers of $xx'$, reducing all
  exponents in $x$ modulo $m$, and all exponents in $x'$ modulo $m'$.  We
  then reduce the result modulo both $\Phi_m(x)$ and $\Phi_{m'}(x')$.  By
  the same argument as above, the cost is softly linear in $\varphi(mm')$.
\end{proof}

\noindent{\bf Extension to several cyclotomic rings.}  The natural
generalization of the algorithm above starts with pairwise distinct
primes $\boldsymbol{p}=(p_1,\dots,p_t)$, non-negative exponent
$\boldsymbol{c}=(c_1,\dots,c_t)$ and variables
$\boldsymbol{x}=(x_1,\dots,x_t)$ over $\F$. Now, we define
$$\H:=\F[x_1,\dots,x_t]/\langle
\Phi_{{p_1}^{c_1}}(x_1),\dots,\Phi_{{p_t}^{c_t}}(x_t)\rangle;$$ when
needed, we will write $\H$ as
$\H_{\boldsymbol{p},\boldsymbol{c},\boldsymbol{x}}$. Finally, we let
$\mu:={p_1}^{c_1}\cdots {p_t}^{c_t}$; then, the dimension $\dim(\H)$ is
$\varphi(\mu)$.

\begin{lemma}\label{lemma:distinctP}
 There exists an $\F$-algebra isomorphism $\Gamma: \H \to
 \F[z]/\langle\Phi_{\mu}(z)\rangle$ given by $x_1 \cdots x_t \mapsto
 z$.  One can apply $\Gamma$ and its inverse to any input using
 $\tilde{O}(\dim(\H))$ operations in $\F$.
\end{lemma}
\begin{proof}
  We proceed iteratively. First, note that the cyclotomic polynomials
  $\Phi_{{p_i}^{c_i}}$ can all be computed in time $O(\varphi(\mu))$. 
  The isomorphism
  $\gamma: \F[x_1,x_2]/\langle \Phi_{{p_1}^{c_1}}(x_1),
  \Phi_{{p_2}^{c_2}}(x_2)\rangle \to \F[z]/\langle
  \Phi_{{p_1}^{c_1}{p_2}^{c_2}}(z)\rangle$
given in the previous paragraph extends coordinate-wise to an
  isomorphism
  $$\Gamma_1: \H \to \F[z,x_3,\dots,x_t]/\langle
  \Phi_{{p_1}^{c_1}{p_2}^{c_2}}(z),\Phi_{{p_3}^{c_3}}(x_3),\dots,\Phi_{{p_t}^{c_t}}(x_t)\rangle.$$
  By the previous lemma, $\Gamma_1$ and its inverse can be applied to
  any input in time $\tilde{O}(\varphi(\mu))$. Iterate this process
  another $t-2$ times, to obtain $\Gamma$ as a product
  $\Gamma_{t-1} \circ \cdots \circ \Gamma_1$. Since $t$ is logarithmic 
  in $\varphi(\mu)$, the proof is complete.
\end{proof}

\noindent{\bf Tensor product of two prime-power cyclotomic rings, same $p$.}~In the following two paragraphs, we discuss the opposite situation as
above: we now work with cyclotomic polynomials of prime power
orders for a common prime $p$. As above, we start with two such polynomials.

Let thus $p$ be a prime. The key to the following algorithms is the
lemma below.  Let $c,c'$ be positive integers, with $c \ge
c'$, and let $x,y$ be indeterminates over $\F$. Define
\begin{align}
\mathbbm{a}&:=\F[x]/\Phi_{p^c}(x),  \\
\mathbbm{b}&:=\F[x,y]/\langle \Phi_{p^c}(x), \Phi_{p^{c'}}(y)\rangle = \mathbbm{a}[y]/\Phi_{p^{c'}}(y).
\end{align}
Note that $\mathbbm{a}$ and $\mathbbm{b}$ have respective dimensions
$\varphi(p^c)$ and $\varphi(p^c) \varphi(p^{c'})$.
\begin{lemma}
  There is an $\F$-algebra isomorphism $\theta: \mathbbm{b} \to
  \mathbbm{a}^{\varphi(p^{c'})}$ such that one can apply $\theta$ or
  its inverse to any inputs using $\tilde{O}(\dim(\mathbbm{b}))$ operations in $\F$.
\end{lemma}
\begin{proof}
  Let $\xbar$ be the residue class of
  $x$ in $\A$. Then, in $\mathbbm{a}[y]$, $\Phi_{p^{c'}}(y)$ factors as
  $$\Phi_{p^{c'}}(y) =\prod_{\substack{1 \le i\le p^{c'}-1\\ \gcd(i,p)
      =1}} (y-\rho_i),$$ with $\rho_i:={\xbar}^{i p^{c-c'}}$ for all
  $i$.  Even though $\mathbbm{a}$ may not be a field, the Chinese
  Remainder theorem implies that $\mathbbm{b}$ is isomorphic to
  $\mathbbm{a}^{\varphi(p^{c'})}$; the isomorphism is given by
  $$\begin{array}{cccc}
    \theta: & \mathbbm{b} & \to & \mathbbm{a} \times \cdots \times \mathbbm{a}, \\
    & P & \mapsto& (P(\xbar,\rho_1),\dots,P(\xbar,\rho_{\varphi(p^{c'})}).
  \end{array}$$
  In terms of complexity, arithmetic operations $(+,-,\times)$ in
  $\mathbbm{a}$ can all be done in $\tilde{O}(\varphi(p^c))$ operations
  in $\F$. Starting from $\rho_1 \in \mathbbm{a}$, all other roots
  $\rho_i$ can then be computed in $O(\varphi(p^{c'}))$ operations in
  $\mathbbm{a}$, that is, $\tilde{O}(\dim(\mathbbm{b}))$
  operations in $\F$. 
  
Applying $\theta$ and its inverse is done by means of fast evaluation
and interpolation~\citep[Chapter~10]{vzGathen13} in $\tilde{O}(\varphi(p^{c'}))$
operations in $\mathbbm{a}$, that is, $\tilde{O}(\deg(\mathbbm{b}))$ operations in $\F$
(the algorithms do not require that $\mathbbm{a}$ be a field).
\end{proof}

\smallskip\noindent{\bf Extension to several cyclotomic rings.}
Let $p$ be as before, and consider now non-negative integers
$\boldsymbol{c}=(c_1,\dots,c_t)$ and variables $\boldsymbol{x}=(x_1,\dots,x_t)$. We
define the $\F$-algebra
$$\A:=\F[x_1,\dots,x_t]/\langle \Phi_{p^{c_1}}(x_1), \dots,
\Phi_{p^{c_t}}(x_t)\rangle,$$ which we will sometimes write
$\A_{p,\boldsymbol{c},\boldsymbol{x}}$ to make the dependency on $p$
and the $c_i$'s clear. Up to reordering the $c_i$'s, we can assume
that $c_1 \ge c_i$ holds for all $i$, and define as before
$\mathbbm{a}:=\F[x_1]/\Phi_{p^{c_1}}(x_1)$.

\begin{lemma}\label{lemma:A}
  There exists an $\F$-algebra isomorphism $\Theta: \A \to
  \mathbbm{a}^{\dim(\A)/\dim(\mathbbm{a})}$. This isomorphism and its
  inverse can be applied to any inputs using $\tilde{O}(\dim(\A))$
  operations in $\F$.
\end{lemma}
\begin{proof}
Without loss of generality, we can assume that all $c_i$'s are non-zero
(since for $c_i=0$, $\Phi_{p^{c_i}}(x_i)=x_i-1$,
so $\F[x_i]/\langle \Phi_{p^{c_i}}(x_i) \rangle = \F$).
We proceed iteratively. First, rewrite $\A$ as
$$\A=\mathbbm{a}[x_2,x_3,\dots,x_t]/\langle \Phi_{p^{c_2}}(x_2), \Phi_{p^{c_3}}(x_3), \dots,
\Phi_{{p_t}^{c_t}}(x_t)\rangle.$$ 
The isomorphism 
$\theta: \mathbbm{a}[x_2]/\Phi_{p^{c_2}}(x_2) \to \mathbbm{a}^{\varphi(p^{c_2})}$
introduced in the previous paragraph extends coordinate-wise
to an isomorphism 
$$\Theta_1: \A \to (\mathbbm{a}[x_3,\dots,x_t]/\langle
\Phi_{p^{c_3}}(x_3), \dots,
\Phi_{p^{c_t}}(x_t)\rangle)^{\varphi(p^{c_2})};$$ $\Theta_1$ and its
inverse can be evaluated in quasi-linear time $\tilde{O}(\dim(\A))$.
We now work in all copies of $\mathbbm{a}[x_3,\dots,x_t]/\langle
\Phi_{p^{c_3}}(x_3), \dots, \Phi_{p^{c_t}}(x_t)\rangle$ independently,
and apply the procedure above to each of them. Altogether we have
$t-1$ such steps to perform, giving us an isomorphism
$$\Theta = \Theta_{t-1} \circ \cdots \circ \Theta_1:
\A \to
\mathbbm{a}^{\varphi(p^{c_2}) \cdots \varphi(p^{c_t})}.$$
The exponent can be rewritten as $ \dim(\A)/\dim(\mathbbm{a})$, as claimed.
In terms of complexity, all $\Theta_i$'s and their inverses can be computed
in quasi-linear time $\tilde{O}(\dim(\A))$, and we do $t-1$ of them,
where $t$ is $O(\log(\dim(\A)))$. 
\end{proof}

\noindent{\bf Decomposing certain $p$-group algebras.}  The prime $p$
and indeterminates $\boldsymbol{x}=(x_1,\dots,x_t)$ are as before; we now consider
positive integers $\boldsymbol{b}=(b_1,\dots,b_t)$, and the $\F$-algebra
\[
\begin{array}{ccl}
\B&:=&\F[x_1,\dots,x_t]/\langle x_1^{p^{b_1}}-1,\dots,x_t^{p^{b_t}}-1\rangle\\$$
&=& \F[x_1]/\langle x_1^{p^{b_1}}-1 \rangle \otimes \cdots \otimes \F[x_t]/\langle x_t^{p^{b_t}}-1 \rangle.
\end{array}
\]
If needed, we will write $\B_{p,\boldsymbol{b},\boldsymbol{x}}$ to make the dependency
on $p$ and the $b_i$'s clear. This is the $\F$-group algebra
of $\Z/p^{b_1}\Z \times \cdots \times \Z/p^{b_t}\Z$.

\begin{lemma}\label{lemma:alg}
  There exists a positive integer $N$, non-negative integers
  $\boldsymbol{c}=(c_1,\dots,c_N)$ and  an
  $\F$-algebra isomorphism 
  $$\Lambda: \B \to \D= \F[z]/\langle \Phi_{p^{c_1}}(z) \rangle \times \cdots \times \F[z]/\langle \Phi_{p^{c_N}}(z)\rangle.$$
  One can apply the isomorphism and its inverse to any 
  input using $\tilde{O}(\dim(\B))$ operations in $\F$.
\end{lemma}
\begin{proof}
For $i \le t$, we have the factorization
$$x_i^{p^{b_i}}-1 = \Phi_1(x_i) \Phi_p(x_i) \Phi_{p^2}(x_i) \cdots
\Phi_{p^{b_i}}(x_i);$$ note that $\Phi_1(x_i)=x_i-1$.  The factors may
not be irreducible, but they are pairwise coprime, so that we have a
Chinese Remainder isomorphism
\[
  \lambda_i: \F[x_i]/\langle x_i^{p^{b_i}}-1 \rangle \to \F[x_i]/\langle \Phi_1(x_i)\rangle
  \times \cdots \times  \F[x_i]/\langle \Phi_{p^{b_i}}(x_i)\rangle.
\]
Together with its inverse, this can be computed  
in $\tilde{O}(p^{b_i})$ operations in $\F$~\citep[Chapter~10]{vzGathen13}. By distributivity of the tensor
product over direct products, 
this gives an $\F$-algebra isomorphism
$$\lambda: \B \to \prod_{c_1=0}^{b_1} \cdots \prod_{c_t=0}^{b_t} \A_{p,\boldsymbol{c},\boldsymbol{x}},$$
with $\boldsymbol{c}=(c_1,\dots,c_t)$. Together with its inverse, 
$\lambda$ can be computed in $\tilde{O}(\dim(\B))$ operations in $\F$.
Composing with the result in Lemma~\ref{lemma:A}, this gives
us an isomorphism
$$\Lambda: \B \to \D:=\prod_{c_1=0}^{b_1} \cdots \prod_{c_t=0}^{b_t}
\mathbbm{a}_{\boldsymbol{c}}^{D_{\boldsymbol{c}}},$$ where
$\mathbbm{a}_{\boldsymbol{c}} = \F[z]/\langle \Phi_{p^c}(z)\rangle$,
with $c =\max(c_1,\dots,c_t)$ and $D_{\boldsymbol{c}} =
\dim(\A_{t,\boldsymbol{c},\boldsymbol{x}})/\dim(\mathbbm{a}_{\boldsymbol{c}})$. As
before, $\Lambda$ and its inverse can be computed in quasi-linear time
$\tilde{O}(\dim(\B))$.
\end{proof}
As for $\B$, we will write $\D_{p,\boldsymbol{b},\boldsymbol{x}}$ if needed; it is
well-defined, up to the order of the factors.

\smallskip

\noindent{\bf Main result.} Let $G$ be an abelian group.  We can write
the elementary divisor decomposition of $G$ as $G = G_1 \times \cdots
\times G_s$, where each $G_i$ is of prime power order $p_i^{a_i}$, for
pairwise distinct primes $p_1,\dots,p_s$, so that $n = |G|$ writes $n
= p_1^{a_1} \cdots p_s^{a_s}$. Each $G_i$ can itself be written as a
product of cyclic groups, $G_i = G_{i,1} \times \cdots \times
G_{i,t_i}$, where the factor $G_{i,j}$ is cyclic of order
${p_i}^{b_{i,j}}$, with $b_{i,1} \le \cdots \le b_{i,t_i}$; this is
the invariant factor decomposition of $G_i$, with $b_{i,1} + \cdots +
b_{i,t_i} = a_i$.

We henceforth assume that generators
$\gamma_{1,1},\dots,\gamma_{s,t_s}$ of respectively
$G_{1,1},\dots,G_{s,t_s}$ are known, and that elements of $\F[G]$ are
given on the power basis in $\gamma_{1,1},\dots,\gamma_{s,t_s}$. Were
this not the case, given arbitrary generators $g_1,\dots,g_r$ of $G$, with
orders $e_1,\dots,e_r$, a brute-force solution would factor each $e_i$
(factoring $e_i$ takes $o(e_i)$ bit operations on a standard RAM), so
as to write $\langle g_i \rangle$ as a product of cyclic groups of
prime power orders, from which the required decomposition follows.

\begin{proposition}
  Given $\beta \in \F[G]$, written on the power basis
  $\gamma_{1,1},\dots,\gamma_{s,t_s}$, one can test if $\beta$ is a
  unit in $\F[G]$ using $\tilde{O}(n)$ operations in $\F$.
  If it is the case, given $\eta$ in $\F[G]$, one can compute
  $\beta^{-1} \eta$ in the same asymptotic runtime.
\end{proposition}
In view of
Lemma~\ref{Lem:Proj-bis}, Proposition~\ref{prop:polycyclic} and the
claim on the cost of invertibility testing prove the first part of
Theorem \ref{thm:main}; the second part of this proposition will allow
us to prove Theorem~\ref{thm:main2} in the next section.

\smallskip

The proof of the proposition occupies the rest of this paragraph.
From the factorization $G = G_1 \times \cdots \times G_s$, we deduce
that the group algebra $\F[G]$ is the tensor product $\F[G_1] \otimes
\cdots \otimes \F[G_s]$. Furthermore, the factorization $G_i = G_{i,1}
\times \cdots \times G_{i,t_i}$ implies that $\F[G_i]$ is isomorphic,
as an $\F$-algebra, to
$$\F[x_{i,1},\dots,x_{i,t_i}]/\left \langle
x_{i,1}^{p_i^{b_{1}}}-1,\dots,x_{i,t_i}^{p_i^{b_{i,t_i}}}-1\right\rangle
=\B_{p_i,\boldsymbol{b}_i,\boldsymbol{x}_i},$$ with $\boldsymbol{b}_i
= (b_{i,1},\dots,b_{i,t_i})$ and $\boldsymbol{x}_i =
(x_{i,1},\dots,x_{i,t_i})$. Given $\beta$ on the power basis in
$\gamma_{1,1},\dots,\gamma_{s,t_s}$, we obtain its image $B$ in
$\B_{p_1,\boldsymbol{b}_1,\boldsymbol{x}_1} \otimes \cdots \otimes
\B_{p_s,\boldsymbol{b}_s,\boldsymbol{x}_s}$ simply by renaming
$\gamma_{i,j}$ as $x_{i,j}$, for all $i,j$.

For $i \le s$, by Lemma~\ref{lemma:alg}, there exist integers
$c_{i,1},\dots,c_{i,N_i}$ such that
$\B_{p_i,\boldsymbol{b}_i,\boldsymbol{x}_i}$ is isomorphic to an
algebra $\D_{p_i, \boldsymbol{b}_i, z_i}$, with factors 
$\F[z_i]/\langle \Phi_{{p_i}^{c_{i,j}}}(z_i) \rangle$.
By distributivity of the tensor product over direct products, we
deduce that $\B_{p_1,\boldsymbol{b}_1,\boldsymbol{x}_1} \otimes \cdots
\otimes \B_{p_s,\boldsymbol{b}_s,\boldsymbol{x}_s}$ is isomorphic to
the product of algebras
 \begin{equation}\label{eq:prod}
\text{\small $\prod$}_{\boldsymbol{j}}~ \F[z_1,\dots,z_s]/
\langle \Phi_{{p_1}^{c_{1,j_1}}}(z_1),\dots, \Phi_{{p_s}^{c_{s,j_s}}}(z_s) \rangle,   
 \end{equation}
for indices $\boldsymbol{j}=(j_1,\dots,j_s)$, with
$j_1 =1,\dots,N_1,\dots,j_s=1,\dots,N_s$;
call $\Gamma$ the isomorphism. Given $B$ in $\B_{p_1,\boldsymbol{b}_1,\boldsymbol{x}_1} \otimes
\cdots \otimes \B_{p_s,\boldsymbol{b}_s,\boldsymbol{x}_s}$,
Lemma~\ref{lemma:alg} also implies that $B':=\Gamma(B)$ can be
computed in softly linear time $\tilde{O}(n)$ (apply the isomorphism
corresponding to $\boldsymbol{x}_1$ coordinate-wise with respect to
all other variables, then deal with $\boldsymbol{x}_2$, etc).
The codomain in~\eqref{eq:prod} is the product of all $\H_{\boldsymbol{p},\boldsymbol{c}_{\boldsymbol{j}},\boldsymbol{z}}$,
with 
$$\boldsymbol{p}=(p_1,\dots,p_s),\quad \boldsymbol{c}=(c_{1,j_1},\dots,c_{s,j_s}),\quad \boldsymbol{z}=(z_1,\dots,z_s).$$
Apply Lemma~\ref{lemma:distinctP} to all 
$\H_{\boldsymbol{p},\boldsymbol{c}_{\boldsymbol{j}},\boldsymbol{z}}$ to obtain
an $\F$-algebra isomorphism
$$\Gamma': \text{\small $\prod$}_{\boldsymbol{j}}~
\H_{\boldsymbol{p},\boldsymbol{c}_{\boldsymbol{j}},\boldsymbol{z}} \to
\text{\small $\prod$}_{\boldsymbol{j}} ~\F[z]/\langle
\Phi_{d_{\boldsymbol{j}}}(z) \rangle,$$ for certain integers
$d_{\boldsymbol{j}}$. The lemma implies that given $B'$,
$B'':=\Gamma'(B')$ can be computed in softly linear time
$\tilde{O}(n)$ as well. Invertibility of $\beta \in \F[G]$ is
equivalent to $B''$ being invertible, that is, to all its components
being invertible in the respective factors $\F[z]/\langle
\Phi_{d_{\boldsymbol{j}}}(z) \rangle$. Invertibility in such an
algebra can be tested in softly linear time by applying the fast
extended GCD algorithm~\citep[Chapter~11]{vzGathen13}, so the first
part of the proposition follows.

Given $\eta$ in $\F[G]$, we can similarly compute its image $H''$ in
$\text{\small $\prod$}_{\boldsymbol{j}} ~\F[z]/\langle
\Phi_{d_{\boldsymbol{j}}}(z) \rangle,$ with the same asymptotic
runtime as for $\beta$. If we suppose $\beta$ (and thus $B''$)
invertible, division in each $\F[z]/\langle
\Phi_{d_{\boldsymbol{j}}}(z) \rangle$ takes softly linear time in the
degree $\phi_{d_{\boldsymbol{j}}}$; as a result, we obtain ${B''}^{-1}
H''$ in time $\tilde{O}(n)$.  One can finally invert all
isomorphisms we applied, in order to recover $\beta^{-1} \eta$ in
$\F[G]$; this also takes time $\tilde{O}(n)$. Summing all costs,
this establishes the second part of the proposition.

%%%%%%%%%%%%%%%%%%%%%%%%%%%%%%%%%%%%%%%%%%%%%%%%%%%%%%%%%%%%

\subsection{Metacyclic Groups}

In this subsection, we study the invertibility and division problems
for a metacyclic group $G$. A group $G$ is metacyclic if it has a
normal cyclic subgroup $H$ such that $G/H$ is cyclic: this is the case
$r=2$ in the definition we gave of polycyclic groups. For instance,
any group with a squarefree order is metacyclic, see
\citep[p.~88]{Johnson} or \citep[p.~334]{Curtis} for more
background. 

For such groups, we will use a standard specific notation, rather than
the general one introduced in~\eqref{eq:polycyclicgrp} for arbitrary
polycyclic ones: we will write $(\sigma,\tau)$ instead of $(g_1,g_2)$
and $(m,s)$ instead of $(e_1,e_2)$. Then, a metacyclic group $G$ can
be presented~as
\begin{equation}
  \label{eq:metacyclic}
  \langle \sigma,\tau: \sigma^m = 1,  \tau^s = \sigma^t, \tau^{-1}\sigma \tau = \sigma^u \rangle,
\end{equation}
for integers $m,t,u,s$, with $u,t \leq m$ and $u^s = 1 \bmod t$, $ut =
t \bmod m$. For example, the dihedral group
$$D_{2m} = \langle \sigma,\tau: \sigma^m =1, \tau^2 = 1, \tau^{-1}
\sigma \tau = \sigma^{m-1} \rangle, $$ is metacyclic, with
$s=2$. Generalized quaternion groups, which can be presented as
$$Q_m = \langle \sigma,\tau: \sigma^{2m} =1, \tau^2 = \sigma^m,
\tau^{-1} \sigma \tau = \sigma^{2m-1} \rangle,$$ are metacyclic, with
$s=2$ as well. Using the notation of~\eqref{eq:metacyclic}, $n=|G|$ is
equal to $ms$, and all elements in a metacyclic group can be presented
uniquely as either
\begin{equation}\label{pres1}
\{\sigma^i \tau^j,\,\,\, 0\leq i \leq m-1,\ 0\leq j \leq s-1\}  
\end{equation}
or
\begin{equation}\label{pres2}
\{ \tau^j\sigma^i,\,\,\, 0\leq i \leq m-1,\ 0\leq j \leq s-1\}.
\end{equation}
Accordingly, elements in the group algebra $\F[G]$ can be written as 
either 
$$\sum_{\substack{i <m\\ j< s}} c_{i,j} \sigma^i \tau^j
\quad\text{or}\quad \sum_{\substack{i <m\\ j< s}} c'_{i,j} \tau^j
\sigma^i.$$ Conversion between the two representations involves no
operation in $\F$, using the commutation relation $\sigma^k \tau^c =
\tau^c \sigma^{ku^c}$ for $k,c \ge 0$.

To test invertibility in $\F[G]$, a possibility would be to rely on
the Wedderburn decomposition of $\F[G]$, but the structure of group
algebras of metacyclic groups is not straightforward to exploit; see
for instance \citep[\S 47]{Curtis} for algebraically closed $\F$, or,
when $\F=\Q$,~\citep{Decomposition} for dihedral and quaternion
groups. Instead, we will highlight the structure of the multiplication
matrices in $\F[G]$.

Take $\beta$ in $\F[G]$. In eq.~\eqref{eqdef:M}, we introduced the
matrix $\mat M_\F(\beta)$ of left multiplication by $\beta$ in
$\F[G]$, where columns and rows were indexed using an arbitrary ordering
of the group elements. We will now reorder the rows and columns of
$\mat M_\F(\beta)$ using the two presentations of $G$ seen in
\eqref{pres1} and~\eqref{pres2}, in order to highlight its block structure.
In what follows, for non-negative integers $a,b,c$, we will write
$\beta_{a,b,c}$ for the coefficient of $\tau^a \sigma^b \tau^c$ in the
expansion of $\beta$ on the $\F$-basis of $\F[G]$.

We first rewrite $\mat M_\F(\beta)$ by reindexing its columns
by 
$$\begin{bmatrix}
(\sigma^0 \tau^0)^{-1} &  \cdots & (\sigma^{m-1} \tau^0)^{-1} \cdots& (\sigma^0 \tau^{s-1})^{-1} &  \cdots & (\sigma^{m-1} \tau^{s-1})^{-1} 
\end{bmatrix}$$
and its rows by
$$\begin{bmatrix}
\tau^0 \sigma^0 & \cdots & \tau^0 \sigma^{m-1} & \cdots& \tau^{s-1} \sigma^0 & \cdots&  \tau^{s-1} \sigma^{m-1}
\end{bmatrix}.
$$ This matrix displays a $s \times s$
block structure. Each block has itself size $m \times m$; for $1 \le
u,v \le s$ and $1 \le a,b \le m$, the entry of index $(a,b)$ in the
block of index $(u,v)$ is the coefficient of $\tau^u \sigma^a \sigma^b
\tau^v$ in $\beta$, that is, $\beta_{u,a+b,v}$. In other words,
all blocks are Hankel matrices.

Using the algorithm of~\cite{BoJeMoSc17} (see
also~\citealt[Appendix~A]{EbGiGiSVi07}), this structure allows us to
solve a system such as $\mat M_\F(\beta) \boldsymbol x = \boldsymbol
y$ in Las Vegas time $\tilde{O}(s^{\omega-1} n)$ (or raise an error if
there is no solution). In addition, if the right-hand side is zero and
$\mat M_\F(\beta)$ is not invertible, the algorithm returns a non-zero
kernel element.
This last remark allows us to test whether $\beta$ is invertible in
Las Vegas time $\tilde{O}(s^{\omega-1} n)$; if so, given the
coefficient vector $\boldsymbol y$ of some $\eta$ in $\F[G]$, we can
compute $\beta^{-1} \eta$ in the same asymptotic runtime.

It is also possible to reorganize the rows and columns of 
 $\mat M_\F(\beta)$, using indices
$$\begin{bmatrix}
(\tau^0 \sigma^0)^{-1} & \cdots & (\tau^0 \sigma^{m-1})^{-1} & \cdots& (\tau^{s-1} \sigma^0)^{-1} & \cdots&  (\tau^{s-1} \sigma^{m-1})^{-1}
\end{bmatrix}$$
for its columns and 
$$\begin{bmatrix}
\sigma^0 \tau^0 &  \cdots & \sigma^{m-1} \tau^0 \cdots& \sigma^0 \tau^{s-1} &  \cdots & \sigma^{m-1} \tau^{s-1}
\end{bmatrix}
$$ for its rows. The resulting matrix has an $m \times m$ block
structure, where each $s\times s$ block is Hankel. As a result, it
allows us to solve the problems above, this time using
$\tilde{O}(m^{\omega-1} n)$ operations in $\F$.  Since we have either
$s\le \sqrt n$ or $m \le \sqrt n$, this implies the following.

\begin{proposition}
  Given $\beta \in \F[G]$, one can test if $\beta$ is a unit in
  $\F[G]$ using $\tilde{O}(n^{(\omega+1)/2})$ operations in $\F$.  If
  it is the case, given $\eta$ in $\F[G]$, one can compute $\beta^{-1}
  \eta$ in the same asymptotic runtime.
\end{proposition}
Combined with Proposition~\ref{prop:polycyclic}, the former statement
provides the last part of the proof of Theorem \ref{thm:main}.

%%% Local Variables:
%%% mode: latex
%%% TeX-master: "NormalBasisCharZero"
%%% End:
 \section{Basis Conversion}\label{sec:conversion}

We conclude this paper with algorithms for basis conversion: assuming
we know that $\alpha$ is normal, we show how to perform the
change-of-basis between the power basis of $\K/\F$ and the normal
basis $G\cdot \alpha$. The techniques used below are inspired by 
those used by~\cite[Section~4]{KalSho98} in the case of extensions
of finite fields.

\subsection{From normal to power basis} Suppose $G = \lbrace g_1, \ldots, g_n
\rbrace$, $\alpha$ is a normal element of $\K/\F$ and we are given $u\in \K$ 
as $u=\sum_{i=1}^n u_i g_i(\alpha)$. In order to write
$u $ in the power basis, we have to compute the matrix-vector 
product
\begin{equation}\label{eq:norm2pwrmat}
\left[\begin{array}{ccc}
\boldsymbol{ \gamma_1} & \cdots & \boldsymbol{ \gamma_n }
\end{array}\right]\cdot 
\begin{bmatrix}
u_1 \\ \vdots \\ u_n
\end{bmatrix},
\end{equation}
where for $i=1,\dots,n$, $\boldsymbol \gamma_i \in \F^{n\times 1}$ is
the coefficient vector of $g_i(\alpha)$. As already pointed out by
Kaltofen and Shoup for finite fields, this shows that conversion from normal to power
basis is the transpose problem of computing the ``projected'' orbit
sum $s_{\alpha,\ell}$, which we solved in Section \ref{sec:osum}. 

The transposition principle then allows us to derive runtime estimates
for the conversion problem; below, we present an explicit procedure
derived from the algorithm in Subsection~\ref{ssec:proj_abelian}.  As
in that section, we give the algorithm in the general case of a
polycyclic group $G$ presented as
\[
G =\{ g_r^{i_r} \cdots g_1^{i_1}, \text{~with~}  0 \leq i_j < e_j \text{~for~} 1 \leq j \leq r\}.
\]
With indices $i_1,\dots,i_r$ as above, we are given a family of 
coefficients $u_{i_1,\dots,i_r}$ in $F$, and we expand the sum
$u=\sum_{i_1,\dots,i_r} u_{i_1,\dots,i_r} g_r^{i_r} \cdots g_1^{i_1}(\alpha)$
on the power basis of $\K/\F$. For this, we let $z \in\{1,\dots,r\}$ be the index
defined in Subsection~\ref{ssec:proj_abelian}.

\smallskip\noindent \textbf{Step 1.} Apply Lemma \ref{lem:selfcomp},
with $\alpha_{i_1,\dots,i_r} = \alpha$ for all $i_1,\dots,i_r$, to get
$$ G_{i_z,\dots,i_1}=g_z^{i_z} \cdots
g_1^{i_1}(\alpha),$$ for all indices $i_1,\dots,i_z$ such that $0\leq
i_m < e_m$ holds for $m=1,\dots,z-1$ and $0\leq i_z <s_z=
\lceil{\sqrt{n}}/(e_1 \cdots e_{z-1})\rceil$.  As in
Subsection~\ref{ssec:proj_abelian}, the cost of this step is
$\osumcosttilde$.

\smallskip\noindent\textbf{Step 2.} Compute $G_z=g_z^{s_z}$, for $s_z$
as above. The cost is is negligible compared to the cost of the
previous step.

\smallskip\noindent\textbf{Step 3.} Compute the matrix product 
$\mat U \mat \Gamma$, where
\begin{itemize}
\item[$\bullet$] $\mat U$ is the matrix over $\F$ having $\lceil
  e_z/s_z\rceil e_{z+1} \cdots e_r$ rows and $e_1 \cdots e_{z-1} s_z$
  columns built as follows. Rows are indexed by $(j_z,\dots,j_r)$,
  with $0\le j_z < \lceil e_z/s_z\rceil$ and $0 \le j_m < e_m$ for all
  other indices; columns are indexed by $(i_1,\dots,i_z)$, with $0\le
  i_z < s_z$ and $0 \le i_m < e_m$ for all other indices; the entry at
  rows $(j_z,\dots,j_r)$ and column $(i_1,\dots,i_z)$ is
  $u_{i_1,\dots,i_z + s_z j_z, j_{z+1},\dots,j_r}$.
\item[$\bullet$] $\mat \Gamma$ is the matrix with $e_1 \cdots e_{z-1}
  s_z$ rows (indexed in the same way as the columns of $\mat U$) and $n$
  columns, whose row of index $(i_1,\dots,i_z)$ contains the
  coefficients of $ G_{i_z,\dots,i_1}$ (on the power basis of $\K$)
\end{itemize}
As established in Subsection~\ref{ssec:proj_abelian}, the row and column
dimensions of $\mat U$ are $O(\sqrt n)$, so this product can 
be computed in $O(n^{(1/2)\cdot\omega(2)})$ operations in $\F$. The rows of 
the resulting matrix give the coefficients of 
$$ H_{j_{z+1},\dots,j_r} = \sum_{i_1,\dots,i_z} u_{i_1,\dots,i_z+s_z j_z,\dots,j_r} g_z^{i_z} \cdots
g_1^{i_1}(\alpha),$$
for all indices $(j_z,\dots,j_r)$ and $(i_1,\dots,i_z)$ as above.

\smallskip\noindent\textbf{Step 4.} Compute and add all
$$g_r^{j_1} \cdots g_{z+1}^{j_{z+1}} G_z^{j_z} ( H_{j_{z+1},\dots,j_r}
),$$ for indices $(j_z,\dots,j_r)$ as above; their sum is precisely
the input element $u=\sum_{i_1,\dots,i_r} u_{i_1,\dots,i_r} g_r^{i_r}
\cdots g_1^{i_1}(\alpha)$, written on the power basis.

This is done by a second call to Lemma \ref{lem:selfcomp}, for the same
asymptotic cost as in Step 1. Summing all costs, we arrive
at an overall runtime of $\osumcosttilde$ operations in $\F$ for
the conversion from normal to power basis. This proves the first 
half of Theorem~\ref{thm:main2}.

%%%%%%%%%%%%%%%%%%%%%%%%%%%%%%%%%%%%%%%%%%%%%%%%%%%%%%%%%%%%

\subsection{Power basis to normal basis.}

Now assume $u \in \K$ is given in the power basis. Still writing the
elements of $G$ as $g_1,\dots,g_n$, the goal is to find coefficients
$c_i$'s in $\F$ such that
$$ \sum_{i = 1}^{n} c_i g_i(\alpha) = u.$$ Starting from this
equality, for any element $g_j$ of $G$, we have
$$ \sum_{i = 1}^{n} c_i g_jg_i(\alpha) = g_j(u).$$
Then, if $\ell$ is a random $\F$-linear projection $\K\to\F$, we get 
$$\sum_{i = 1}^{n} c_i \ell(g_jg_i(\alpha)) = \ell(g_j(u)), \quad  1 \le j \le n.$$
Introducing 
$$u'= \sum_{i=1}^{n} c_i g_i^{-1} \in \F[G]$$
and writing as before
$$s_{\alpha,\ell}=\sum_{j=1}^{n} \ell(g_j(\alpha))g_j
\quad\text{and}\quad
s_{u,\ell} = \sum_{j=1}^{n}\ell(g_j(u)) g_j \quad\text{in~} \F[G],$$ the
$n$ equations above are equivalent to the equality
$s_{\alpha,\ell}\ u' = s_{u,\ell}$ in $\F[G]$.

We use the algorithm of Section~\ref{sec:osum} to compute both
$s_{\alpha,\ell}$ and $s_{u,\ell}$; this takes $\osumcosttilde$
operations in $\F$, for $G$ polycyclic. If $\alpha$ is normal,
$s_{\alpha,\ell}$ is a unit for a generic $\ell$. Then, if we further
assume that $G$ is either abelian or metacyclic, it suffices to apply
the division algorithms given in the previous section to recover $u'$,
and thus all coefficients $c_1,\dots,c_n$. In both cases, the runtime
of the division is negligible compared to the cost $\osumcosttilde$ of
the first step. Altogether, this finishes the proof of
Theorem~\ref{thm:main2}.

%%% Local Variables:
%%% mode: latex
%%% TeX-master: "NormalBasisCharZero"
%%% End:
 
\newcommand{\Gathen}{\relax}

 % Put your bibliography into the above
                                % filecontents environment!

\end{document}